\documentclass{IET}

\usepackage{graphicx,epsfig,epstopdf,amsmath,amssymb,enumerate} 
\usepackage{cite}  
\usepackage{algpseudocode}
\usepackage{amsthm}
\newtheorem{theorem}{Theorem}
\newtheorem{lemma}[theorem]{Lemma}
\newtheorem*{theorem*}{Theorem}
\newtheorem*{lemma*}{Lemma} 

\begin{document}

\title{Single Cell and Multi-cell Performance Analysis of OFDM Index Modulation}

\author[1,*]{Shangbin Wu}
\affil{Samsung R\&D Institute UK}

\author[1]{Maziar Nekovee}

\affil[*]{shangbin.wu@samsung.com}

\abstract{This paper addresses the achievable rate of single cell and sum rate of multi-cell orthogonal frequency division multiplexing (OFDM) index modulation (IM). The single cell achievable rate of OFDM-IM with Gaussian input is calculated using a multi-ary symmetric channel. Then, the cumulative distribution function (CDF) of multi-cell OFDM-IM is investigated by stochastic geometry. Furthermore, it is proved in this paper that the probability density function (PDF) of noise plus intercell-interference (ICI) in multi-cell OFDM-IM with quadrature amplitude modulation (QAM) follows a mixture of Gaussians (MoG) distribution. Next, parameters of the MoG distribution are estimated using a simplified expectation maximization (EM) algorithm. Upper bound of sum rates of multi-cell OFDM-IM is derived. Furthermore, analytic and simulated results are compared and discussed.}

\maketitle

\section{Introduction}
The fifth generation (5G) cellular network is emerging to satisfy the unprecedented growth in data traffic and the number of connected devices. A key performance indicator (KPI) of 5G is the ability to provide smooth quality of user experience at cell edges, which requires above 1 Gbps data rates. Recently, orthogonal frequency division multiplexing (OFDM) index modulation (IM) \cite{Alhiga09} was proposed as one of the 5G enabling technologies, as it has advantages such as increase in per subcarrier transmit power and reduction in inter-cell interference (ICI). It was reported in \cite{Hong14} and \cite{Seol09} that the ICI of legacy OFDM networks follows a Gaussian distribution which caused most throughput degradation. However, the ICI of OFDM-IM is not Gaussian distributed. Therefore, OFDM-IM is able to provide room to optimize data rate at cell edges. The philosophy behind OFDM-IM is that only one of a number of OFDM subcarriers is active when transmitting symbols. In addition to the information carried by the transmitted symbol, the subcarrier index can also be used to convey information. The idea of conveying information via indexes was first proposed in the spatial domain, i.e., spatial modulation (SM) \cite{Mesleh08}. A variant version of OFDM-IM was proposed in \cite{Tsonev11}, where bits were divided into blocks of bits before the OFDM-IM modulator. Authors in \cite{Basar13} studied practical implementation issues of OFDM-IM such as maximum likelihood (ML) detector, log-likelihood ratio (LLR) detector, and impact of channel estimation errors. A low complexity ML detector for OFDM with in-phase/quadrature IM was discussed in \cite{Zheng15}, which was implemented with a priori knowledge of noise variance. In \cite{Xiao14}, additional interleaving was introduced to subcarriers in correlated channels, in order to provide extra diversity gain to the OFDM-IM. Later, the combination of OFDM-IM and SM was proposed in \cite{Datta15}, where the symbol domain, subcarrier domain, and spatial domain formed a three dimensional signal space. Then, the generalized OFDM-IM was introduced in \cite{Fan15}, where multiple subcarriers were active. The information conveyed by the index relies on different combinations of active subcarrier indexes. When the transmitted symbols of OFDM-IM are quadrature-amplitude modulation (QAM) symbols, this type of OFDM-IM was named frequency QAM (FQAM) \cite{Hong14}. Studies of OFDM--IM in \cite{Alhiga09},\cite{Hong14}, \cite{Basar13}--\hspace{-0.01cm}\cite{Datta15} focused on bit error rate (BER) and frame error rate (FER) performance. 

Achievable rate of OFDM--IM with different settings were reported in \cite{Wen15}--\hspace{-0.001cm}\cite{Ishikawa16}. OFDM-IM with finite constellation input was discussed in \cite{Wen15}, where a closed-form lower bound of the achievable rate was derived. The application of OFDM-IM for underwater acoustic communications as well as achievable rate with finite constellation input were investigated in \cite{Wen16}. Although achievable rate of OFDM-IM with Gaussian input was investigated in \cite{Ishikawa16}, closed-form expression was not provided in \cite{Ishikawa16}.

Also, the performance of OFDM-IM in multi-cell scenarios has been less studied. System level simulations (SLSs) over typical hexagonal multi-cell network are able to provide certain insights of multi-cell performance of wireless networks with OFDM-IM. However, there are two drawbacks of SLSs. First, SLSs are time consuming. Second, the hexagonal multi-cell network layout is usually not fulfilled in realistic situations, where base stations (BSs) are approximately distributed in a uniform manner. Therefore, it is beneficial to have analytic results on the performance of multi-cell scenarios. This can be investigated via a mathematics tool called stochastic geometry \cite{Baccelli_v1}--\hspace{-0.001cm}\cite{Baccelli15}, where BSs were assumed to be distributed following a Poisson Point Process (PPP). In this case, the cumulative distribution function (CDF) of the signal to interference plus noise (SINR) was expressed in closed form in different scenarios, such as ad hoc networks \cite{Baccelli06}, cellular networks \cite{Andrews11}, and cooperative networks \cite{Baccelli15}.

Authors in \cite{Hong14} studied the statistics of ICI in multi-cell OFDM-IM with QAM inputs. The generalized Gaussian distribution was used to approximate the distribution of noise plus ICI. However, the exact distribution of noise plus ICI of multi-cell OFDM-IM with QAM inputs was missing in the literature. In this paper, this exact distribution will be found.

The contributions of this paper are listed as below:
\begin{enumerate}
\item The subcarrier index detection error probability of single cell OFDM-IM with Gaussian input is conducted. Then, closed-form single cell achievable rate of OFDM-IM with Gaussian input is derived. 

\item The CDF of SINR of multi-cell OFDM-IM is derived using stochastic geometry.

\item The distribution of ICI of multi-cell OFDM-IM with QAM input is derived, showing that it follows a mixture of Gaussians (MoG) distribution. In addition, the parameters of the probability density function (PDF) of ICI are computed using a simplified expectation maximization (EM) algorithm. Then, the upper bound of sum rates of multi-cell OFDM-IM with QAM input is studied.

\end{enumerate}

The rest of this paper is structured as follows. Section~\ref{sec_System_Model} gives a general description of the system model. Achievable rate of single cell OFDM-IM with Gaussian input will be investigated in Section~\ref{sec_Single_Link_IM}. Section \ref{sec_Multi_Cell_IM} will study the CDF of SINR of multi-cell OFDM-IM with stochastic geometry. Also, the distribution and its parameters of ICI are analyzed. Simulation results and analysis are presented in Section \ref{sec_results_analysis}. Conclusions are drawn in Section~\ref{sec_conclusion_section}.

\section{System Model} \label{sec_System_Model}
Let us consider a downlink multi-cell network using OFDM-IM, where a target user equipment (UE) is located at the origin. BSs are distributed as a homogeneous PPP with density $\lambda$.
The set of all BSs is denoted as $\mathcal{S}$ and the set of BSs interfering the target UE is denoted as $\mathcal{S}'$. Let $\alpha$ be the pathloss coefficient. Assume that the OFDM-IM system has $N_\mathrm{F}$ subcarriers, then $\log_2 N_\mathrm{F}$ bits are conveyed by subcarrier indexes.
$N_\mathrm{B}$ is the number of BSs.
Let $F$ be the subcarrier index, which is a uniformly distributed random variable defined on $\left\lbrace 1, 2, \cdots,  {N_\mathrm{F}}\right\rbrace$ and let $\mathcal{H}_\xi=\left\lbrace h_{\xi,1}, h_{\xi,2}, \cdots, h_{{\xi,N_{\mathrm{F}}}} \right\rbrace$ be the set of all channel coefficients from the $\xi$th BS to the target user on these subcarriers. The channel coefficients $h_{\xi,k}$ ($1\leqslant k\leqslant N_\mathrm{F}$) are independently and identically distributed (i.i.d.) zero mean unit variance complex Gaussian random variables. In practice, this i.i.d. channel assumption can be achieved by introducing interleaving between subcarriers as \cite{Xiao14}. The interleaving can be done via a pseudo random sequence, which is shared by the BS and UE, such that the BS and UE can map or de--map between the original subcarrier indices and the interleaved subcarrier indices. Throughout the paper, we assume that the target UE has perfect knowledge of the channel coefficients from the associated BS but no knowledge from other BSs. Let the $\xi$th BS be the associated BS of the target UE. The distance between the $\xi$th BS and the target UE is $d$ and the distance between the $l$th BS and the target UE is $d_l$ ($\forall l\neq \xi$). Then, the received signal $Y$ of the target UE can be expressed as
\begin{align}
Y=\sqrt{T_\xi}H_\xi X_\xi+N+I,
\label{equ_rx_signal}
\end{align}
where $N$ is a zero mean complex Gaussian noise with variance $\sigma^2_\mathrm{N}$, $H_\xi$ is uniformly distributed random variable defined on $\mathcal{H}_\xi$, $X_\xi$ is the transmitted symbol from the $\xi$th BS to the target user, $T_\xi$ is the average received power (including transmit power, path loss, and shadow fading) from the $\xi$th BS to the target UE, and $I$ is the interference from other BSs. Thus, interference $I$ can be written as
\begin{align}
I=\sum\limits_{l\neq \xi}\sqrt{T_l}H_lX_l\zeta_l,
\label{equ_interference}
\end{align}
where $\zeta_l=1$ if the $l$th BS is transmitting on the same subcarrier as the $\xi$th BS and $\zeta_l=0$ if the $l$th BS is not transmitting on the same subcarrier as the $\xi$th BS. This is due to a basic property of OFDM-IM, which activates only one subcarrier in each transmission period.

\section{Single Cell OFDM-IM} \label{sec_Single_Link_IM}
Achievable rate of OFDM-IM with QAM input and other finite constellation inputs can be found in \cite{Hong14}, \cite{Wen15}--\hspace{-0.001cm}\cite{Ishikawa16}. However, closed-form expression for  achievable rate of sing-cell OFDM-IM with Gaussian input is missing in the literature. Therefore, to fill this gap, single cell achievable rate of OFDM-IM with Gaussian input is analyzed in this section.  Since single cell is considered, subscript $\xi $ representing the $\xi$th BS is dropped for brevity and the interference term $I$ equals $0$. Information of OFDM-IM is conveyed by the symbol $X$ and the frequency index $F$. The achievable rate $r$ in this paper is defined by the average maximum achievable mutual information between the information source $(X,F)$ and the destination $Y$, which can be characterized as
\begin{align}
r&=\mathrm{E}\left[\max I(X,F;Y) \right]\nonumber\\
&\approx \mathrm{E}\left[\max I(X;Y|F) \right]+\mathrm{E}\left[\max I(F;Y) \right]=r_1+r_2.
\end{align}

Since the channel coefficient $H$ is determined once the subcarrier index $F$ is determined, the achievable rate of the single cell OFDM-IM generated by the symbol in terms of signal to noise ratio (SNR) $\rho$ can be calculated using the achievable rate of the Rayleigh fading channel, i.e.,
\begin{align}
r_1(\rho)&=\mathrm{E}\left[\max I(X;Y|F) \right]=\mathrm{E}\left[\max I(X;Y|H) \right]\nonumber\\
&=\int\limits_0^{\infty}\log_2\left(1+z\rho \right)e^{-z}dz=-\frac{1}{\mathrm{ln}2}\mathrm{Ei}\left(-\frac{1}{\rho} \right)\exp(\frac{1}{\rho}),
\label{equ_r1}
\end{align}
where $\mathrm{Ei}\left(\cdot \right)$ is the exponential integral \cite{TableIntegral} defined by $\mathrm{Ei}\left(-z \right)=-\int^{\infty}_{z}\frac{e^{-t}}{t}dt.$ The SNR $\rho$ can easily be controlled by adjusting the BS transit power to compensate path loss and shadow fading in a single cell scenario.

The calculation of $r_2$ is equivalent to determining how much information is retrieved from $Y$ when the information is conveyed on a certain subcarrier $F$. The information retrieving process is not perfect because of the existence of noise. As a result, the subcarrier index may be incorrectly detected. Let $\hat{F}$ be the detected subcarrier index and $P_{F\neq \hat{F}}$ be the subcarrier index detection error probability and denote $P_{F=k \cap \hat{F}=l}$ as the probability that the $k$th subcarrier is used whereas the $l$th subcarrier is detected. With the assumption that the channel coefficients on subcarriers are i.i.d., it can be observed that $P_{F=f_k \cap \hat{F}=f_l}=\frac{P_{F\neq \hat{F}}}{N_\mathrm{F}-1}$ ($\forall l\neq k$).
\begin{figure} 
\centering\includegraphics[width=4in]{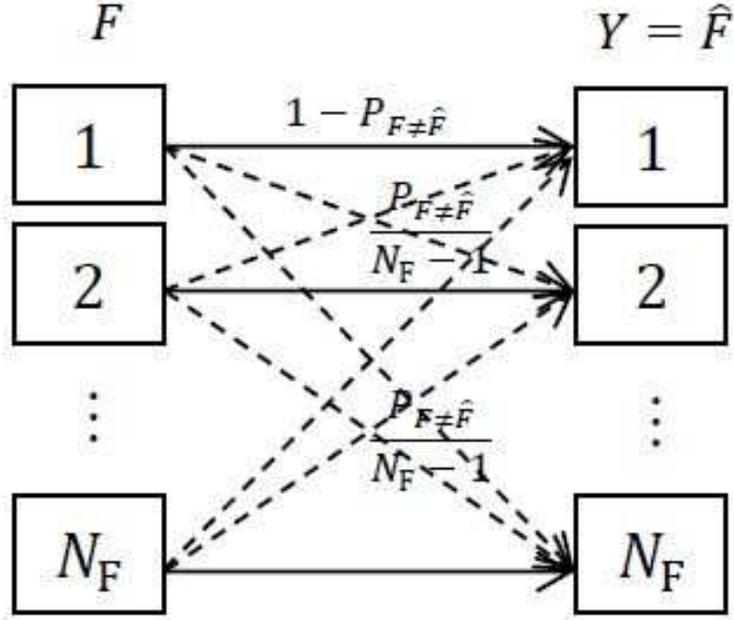} 
\centering\caption{Diagram of an $N_\mathrm{F}$-ary symmetric channel.}
\label{fig_Symmetric_channel}
\end{figure}
Hence, the channel between $F$ and $Y$ can be abstracted by a $N_\mathrm{F}$-ary symmetric channel as depicted in Fig. \ref{fig_Symmetric_channel}. The achievable rate $r_2(\rho)$ of a $N_\mathrm{F}$-ary symmetric channel depends on the subcarrier index detection error probability and can be presented as \cite{Weidmann12}
\begin{align}
r_2(\rho)=\log_2N_{\mathrm{F}}-H_b\left[P_{F\neq \hat{F}}(\rho) \right]-P_{F\neq \hat{F}}(\rho)\log_2(N_{\mathrm{F}}-1),
\label{equ_r2}
\end{align}
where $H_b[\mu]$ is the binary entropy function defined by $H_b[\mu]=-\mu\log_2\mu-(1-\mu)\log_2(1-\mu)$. In order to calculate $r_2(\rho)$, we need to compute $P_{F\neq \hat{F}}(\rho)$ with the following lemma.
%
\begin{lemma}{The subcarrier index detection error probability can be presented as} 
\label{Lemma1}
\begin{align}
P_{F\neq \hat{F}}(\rho)=1-\sqrt{\frac{\pi}{2}}\sum\limits_{k=0}^{N_{\mathrm{F}}-1}C^k_{N_{\mathrm{F}}-1}(-1)^{N_{\mathrm{F}}-k-1}\frac{\left[1-\Phi\left(\sqrt{\frac{N_{\mathrm{F}}-k}{2\rho (N_{\mathrm{F}}-k-1)}} \right) \right]}{\sqrt{\rho(N_{\mathrm{F}}-k-1)(N_{\mathrm{F}}-k)}}\exp\left( \frac{N_{\mathrm{F}}-k}{2\rho (N_{\mathrm{F}}-k-1)}\right),
\label{equ_P}
\end{align}
where $C^k_{N_{\mathrm{F}}-1}$ is the binomial coefficient defined by $C^k_{N_{\mathrm{F}}-1}=\frac{(N_{\mathrm{F}}-1)!}{k!(N_{\mathrm{F}}-k-1)}$ and $\Phi(z)$ is the error function defined by $\Phi(z)=\frac{2}{\sqrt{\pi}}\int\limits_0^z e^{-t^2}dt
$.
\end{lemma}
\begin{proof}
Conditioning on $X$, errors occur in the detection of the subcarrier index when the received signal power on the intended subcarrier is less then any one of the rest subcarriers. Due to the symmetry, it is sufficient to calculate the subcarrier index detection error probability assuming that the first subcarrier is used. Let $Y_k$ ($1\leqslant k \leqslant N_\mathrm{F}$) be the received signal on the $k$th subcarrier. Then, $Y_1=HX+N$ and $Y_2,Y_3,...,Y_{N_\mathrm{F}}$ have the same distribution as $N$. With the condition $X=x$, $|Y_1|^2$ is an exponential random variable with mean $x^2+\sigma_{\mathrm{N}}^2$ and $|Y_{k\neq 1}|^2$ are exponential random variables with mean $\sigma_{\mathrm{N}}^2$. Let $G_{|Y_2|^2}$ be the CDF of $|Y_2|^2$. Then, $P_{F\neq \hat{F}|X=x}(\rho)$ is calculated as 
\begin{align}
P_{F\neq \hat{F}|X=x}(\rho)&=
\Pr\left\lbrace |Y_1|^2<\max\left\lbrace |Y_2|^2,\cdots,|Y_{N_{\mathrm{F}}}|^2 \right\rbrace|X=x \right\rbrace\nonumber\\
&=1-\int\limits_0^\infty p_{|Y_1|^2}(z)G_{|Y_2|^2}^{N_{\mathrm{F}}-1}(z)dz\nonumber\\
&=1-\int\limits_0^\infty\frac{\exp\left\lbrace-\frac{z}{x^2+\sigma_{\mathrm{N}}^2}\right\rbrace}{x^2+\sigma^2_{\mathrm{N}}} \sum\limits_{k=0}^{N_{\mathrm{F}}-1}\begin{pmatrix}
N_{\mathrm{F}}-1
\\ 
k
\end{pmatrix}(-1)^{N_{\mathrm{F}}-k-1} e^{-\frac{1}{\sigma^2_{\mathrm{N}}}(N_{\mathrm{F}}-k-1)z} dz\nonumber\\
&=1-\sum\limits_{k=0}^{N_{\mathrm{F}}-1}\begin{pmatrix}
N_{\mathrm{F}}-1
\\ 
k
\end{pmatrix}(-1)^{N_{\mathrm{F}}-k-1}\int\limits_0^\infty\frac{\exp\left\lbrace-\frac{z}{x^2+\sigma_{\mathrm{N}}^2}-\frac{z}{\sigma^2_{\mathrm{N}}}(N_{\mathrm{F}}-k-1)\right\rbrace}{x^2+\sigma^2_{\mathrm{N}}}   dz\nonumber\\
&=1-\sum\limits_{k=0}^{N_{\mathrm{F}}-1}\begin{pmatrix}
N_{\mathrm{F}}-1
\\ 
k
\end{pmatrix}(-1)^{N_{\mathrm{F}}-k-1}\frac{1}{\frac{x^2}{\sigma^2_{\mathrm{N}}}(N_{\mathrm{F}}-k-1)+N_{\mathrm{F}}-k}\nonumber\\
&=1-\sum\limits_{k=0}^{N_{\mathrm{F}}-1}C^k_{N_{\mathrm{F}}-1}(-1)^{N_{\mathrm{F}}-k-1}\frac{1}{x^2\rho(N_{\mathrm{F}}-k-1)+N_{\mathrm{F}}-k}.
\end{align}
Since $x$ follows a Gaussian distribution with zero mean and variance one, the PDF of $x^2$ is the Gamma distribution with degree of freedom one, i.e., $p_{X^2}(z)=\frac{1}{\sqrt{2\pi z}}e^{-\frac{z}{2}}.$
Unconditioning the frequency index detection error probability, $P_{F\neq \hat{F}}(\rho)$ can be obtained as
\begin{align}
P_{F\neq \hat{F}}(\rho)&=\int\limits_{-\infty}^\infty P_{F\neq \hat{F}|X=x}(\rho)p_X(x)dx \nonumber\\
&=\int\limits_{0}^\infty P_{F\neq \hat{F}|X=z}(\rho)p_{X^2}(z)dz \nonumber\\
&=1-\sum\limits_{k=0}^{N_{\mathrm{F}}-1}\int\limits_0^{\infty}\frac{C^k_{N_{\mathrm{F}}-1}(-1)^{N_{\mathrm{F}}-k-1}}{z\rho(N_{\mathrm{F}}-k-1)+N_{\mathrm{F}}-k}\frac{1}{\sqrt{2\pi z}}e^{-z/2}dz.
\end{align}
Solving the integral \cite{TableIntegral}, (\ref{equ_P}) can be obtained.
\end{proof}

\section{Multi-Cell OFDM-IM} \label{sec_Multi_Cell_IM}
After deriving the single cell sum rate of OFDM-IM with Gaussian inputs in Section \ref{sec_Single_Link_IM}, multi-cell OFDM-IM is investigated in this section. In practical system, finite alphabet inputs are usually used instead of Gaussian inputs. Moreover, the impact of ICI needs to be studied. Therefore, OFDM-IM with QAM inputs is assumed in this section and other types of constellations can be obtained in a similar manner. The CDF of SINR, the PDF of noise plus ICI, and the multi-cell sum rate will be derived.

\subsection{CDF of SINR}
Since only one subcarrier of the target UE and the associated BS is active, the density of interfering BSs to the target UE is one $N_\mathrm{F}$th of the original BS density. Let $\tilde{\rho}$ denote the SINR and $P_T$ denote the transmit power of each BS.
The derivation of the CDF of SINR is directly generalized from \cite{Baccelli06}. The normalized interference power $\mathcal{I}$ from other BSs transmitting on the same subcarrier to the target user can be computed as
\begin{align}
\mathcal{I}=\sum\limits_{l\in \mathcal{S}'}|h_l|^2|d_l|^{-\alpha}.
\end{align}
The channel $h_\xi$ between the target user and his associated BS follows Rayleigh distribution. Hence, $|h_\xi|^2$ is an exponentially distributed random variable. The CDF $G_{\tilde{\rho}}(\tilde{\rho})$ of SINR can then be computed as \cite{Baccelli06}
\begin{align}
G_\mathrm{\tilde{\rho}}(\tilde{\rho})&=1-\Pr\left\lbrace \frac{P_T|h_\xi|^2d^{-\alpha}}{\sigma^2_N+P_T\mathcal{I}}\geqslant \tilde{\rho} \right\rbrace\nonumber\\
&=1-\exp\left(-P_T^{-1}d^\alpha\tilde{\rho}\sigma^2_N \right)\mathrm{E}\left[\exp\left(-d^\alpha \tilde{\rho} \mathcal{I} \right)\right].
\label{equ_G_rho}
\end{align}

Using the Laplace transform of the exponential function and the probability generating functional \cite{Baccelli06}, the expectation part in (\ref{equ_G_rho}) can be expressed as
\begin{align}
\mathrm{E}\left[\exp\left(-d^\alpha \tilde{\rho} \mathcal{I} \right)\right]&=\mathrm{E}\left[\exp\left(-z \mathcal{I} \right)\right]_{|z=d^\alpha \tilde{\rho}}\nonumber\\
&=\exp\left(-\frac{\lambda}{N_{\mathrm{F}}}\int_{\mathfrak{R}^2}\frac{1}{1+z^{-1}|y|^\alpha}dx \right)_{|z=d^\alpha \rho}\nonumber\\
&=\exp\left(-\frac{\lambda}{N_{\mathrm{F}}}z^{2/\alpha}\frac{2\pi^2}{\alpha\sin(2\pi/\alpha)} \right)_{|z=d^\alpha \tilde{\rho}}.
\end{align}
Hence, the CDF $G_{\tilde{\rho}}(\tilde{\rho})$ of multi-cell OFDM-IM can be expressed as
\begin{align}
G_\mathrm{\tilde{\rho}}(\tilde{\rho})=1-\exp\left(-P_T^{-1}d^\alpha\tilde{\rho}\sigma^2_N \right)\exp\left(-\frac{\lambda}{N_{\mathrm{F}}}d^2\tilde{\rho}^{\frac{2}{\alpha}}\frac{2\pi^2}{\alpha\sin(2\pi/\alpha)} \right).
\label{equ_G_CDF}
\end{align}
To avoid $0$ at the denominator in (\ref{equ_G_CDF}), $\alpha$ should be larger than $2$. Typical values of $\alpha$ are between $2$ to $4$.
\subsection{PDF of Noise plus Interference}
The exact PDF of noise plus ICI of multi-cell OFDM-IM with QAM inputs remained unanswered in the literature. Generalized Gaussian distributions were used to approximate the PDF in \cite{Hong14} and \cite{Seol09}. In this section, the exact PDF of noise plus ICI will be derived, showing that noise plus ICI is MoG distributed.
Let $Q=2^{N_B-1}$ and denote $\psi$ as the noise plus ICI, i.e., $\psi= N+\sum\limits_{l\neq \xi}\sqrt{T_l}H_lX_l\zeta_l$.

\begin{theorem} \label{theorem1}
The PDF $p_{\psi}$ of noise plus ICI in multi-cell OFDM-IM with QAM inputs follows a MoG distribution, i.e., 
\begin{align}
p_{\psi}(z)=\sum\limits_{k=1}^{Q}\omega_k\mathcal{N}(z;0,\sigma_k^2).
\end{align}

\label{the_PDF_NI}
\end{theorem}
\begin{proof}
In this proof, $4$QAM is used for simplicity. The distribution of one ICI term , i.e. $\psi_l= h_lX_l\zeta_l$, is first computed. It can be easily shown that $ h_lX_l \sim \mathcal{CN}(0,1)$. Also, the PDF of $\zeta_l$ for 4QAM can be expressed as $p_{\zeta_l}(\zeta_l)=\frac{1}{N_\mathrm{F}}\delta(\zeta_l-1)+\frac{N_\mathrm{F}-1}{N_\mathrm{F}}\delta(\zeta_l)$
which is directly obtained from the fact that the subcarrier index is chosen uniformly. Next, according to the product distribution,
\begin{align}
p_{ h_lX_l \zeta_l}(z)&=\int\limits_{-\infty}^{\infty}p_{\zeta_l}(\tau)\frac{1}{{\pi}}\exp(-|z/\tau|^2)\frac{1}{|\tau|}d\tau\nonumber\\
&=\frac{1}{N_\mathrm{F}}{\frac{1}{{{\pi}}}}\exp(-|z|^2)+\frac{N_\mathrm{F}-1}{N_\mathrm{F}}\lim\limits_{\epsilon\rightarrow 0}\frac{1}{{\pi}\epsilon}\exp\left(-\frac{|z|^2}{\epsilon} \right)\nonumber\\
&=\frac{1}{N_\mathrm{F}}{\frac{1}{{{\pi}}}}\exp(-|z|^2)+\frac{N_\mathrm{F}-1}{N_\mathrm{F}}\delta(z).
\end{align}
Hence, the PDF of one interference term is a weighted sum of a Gaussian function and a Dirac delta function. Using the properties of convolution, the PDF of the total interference can be expressed as
\begin{align}
p_{ h_1X_1 \zeta_1}(z)&* \cdots *p_{ h_{ \xi-1}X_{\xi-1} \zeta_{\xi-1}}(z)*p_{ h_{\xi+1}X_{\xi+1} \zeta_{\xi+1}}(z)*\cdots *p_{ h_{N_\mathrm{B}}X_{N_\mathrm{B}} \zeta_{N_\mathrm{B}}}(z)\nonumber\\
&=\sum\limits_{k=1}^{Q-1}\tilde{\omega}_k\mathcal{CN}(z;0,\tilde{\sigma}_k^2)+\tilde{\omega}_{Q}\delta(z),
\end{align}
where $*$ denotes the convolution operator and $\tilde{\omega}_k\geqslant 0$ with $\sum\limits_{k}\tilde{\omega}_k=1$. By adding the noise term, the PDF $p_\psi$ of noise plus ICI can be calculated as
\begin{align}
p_\psi(z)&=\mathcal{CN}(z;0,\sigma_\mathrm{N}^2)*\left[\sum\limits_{k=1}^{Q-1}\tilde{\omega}_k\mathcal{CN}(z;0,\tilde{\sigma}_k^2)+\tilde{\omega}_Q\delta(z)\right]\nonumber\\
&=\sum\limits_{k=1}^{2^{N_B-1}}{\omega}_k\mathcal{CN}(z;0,\sigma_k^2),
\label{equ_p_psi}
\end{align}
where ${\omega}_k\geqslant 0$ with $\sum\limits_{k}{\omega}_k=1$.
\end{proof}
\begin{lemma}
The PDF $p_{Y}$ of the received signal $Y$ in multi-cell OFDM-IM with QAM inputs follows a MoG distribution, i.e.,
\begin{align}
p_{Y}(y)=\sum\limits_{k=1}^{Q}\omega^{\prime}_k\mathcal{CN}(y;0,\sigma_k^{\prime2}).
\end{align}
 
\end{lemma}
\begin{proof}
This can be obtained directly from Theorem \ref{theorem1} and (\ref{equ_rx_signal}), because the desired signal also follows a Gaussian distribution.
\end{proof}

Next, parameters $\mathbf{\Omega}=\left[\omega_1 \quad \omega_2 \cdots \omega_{Q}\right]$ and $\mathbf{v}=\left[\sigma^2_1 \quad \sigma^2_2 \cdots \sigma^2_{Q}\right]$ need to be estimated. It is well known that the EM algorithm has been widely used to estimate parameters of MoG distributions. Detailed derivations of EM algorithm for MoG parameter estimation are beyond the scope of this paper. Interested readers can find details in \cite{Bishop_PRML}. In particular, for FQAM, the traditional EM algorithm \cite{Bishop_PRML} can be further simplified because the means of the channel, the transmitted symbol, the ICI, and noise are all zero. The simplified EM algorithm is presented in Fig. \ref{fig_algorithm}. Assuming samples of the noise plus ICI are measured as $\left\lbrace \tau_a \right\rbrace_{a=1}^{N_s}$. First, $\mathbf{\Omega}$ and $\mathbf{v}$ are randomly chosen as initialization. Second, E step and M step are operated iteratively until a certain convergence condition is met. The number of Gaussian functions in (\ref{equ_p_psi}) grows exponentially with the number of BSs, which is impractical due to high complexity. However, only a small number of BSs are dominating the total interference. In this case, a small number $Q'$ ($Q'<<Q$) of Gaussian functions will be sufficient to approximate the distribution of noise plus ICI. The parameters of the PDF $p_{Y}$ of the received signal, i.e., $\mathbf{\Omega}^{\prime}=\left[\omega^{\prime}_1 \quad \omega^{\prime}_2 \cdots \omega^{\prime}_{Q}\right]$ and $\mathbf{v}=\left[\sigma^{\prime2}_1 \quad \sigma^{\prime2}_2 \cdots \sigma^{\prime2}_{Q}\right]$, can also be estimated via the same procedure.

\begin{figure}
\hrulefill
\begin{algorithmic}[1]
\State Initialize $\mathbf{\Omega}$ and $\mathbf{v}$ randomly
\State E step: Compute 
\begin{align}
\eta_{ak}=\frac{\omega_k\mathcal{CN}(\tau_a;0,\sigma_k^2)}{\sum\limits_{k=1}^Q\omega_k\mathcal{CN}(\tau_a;0,\sigma_k^2)}
\end{align}
\State M step: Update $\sigma_k^{2,\mathrm{new}}$ and $\omega_k^{\mathrm{new}}$   according to 
\begin{align}
\sigma_k^{2,\mathrm{new}}=\frac{1}{\sum\limits_{a=1}^{N_s}\eta_{ak}}\sum\limits_{a=1}^{N_s}\eta_{ak}|\tau_a|^2
\end{align}
\begin{align}
\omega_k^{\mathrm{new}}=\frac{\sum\limits_{a=1}^{N_s}\eta_{ak}}{N_s}. 
\end{align}
\State Repeat E step and M step until convergence condition is met.

\end{algorithmic}
\hrulefill
\caption{EM algorithm for parameter estimation of noise plus ICI in OFDM-IM systems.}
\label{fig_algorithm}
\end{figure}

\subsection{Multi-Cell Sum Rate}
The multi-cell sum rate $r_3$ of OFDM-IM is defined by the mutual information between the received signal and the information source as
\begin{align}
r_3=\mathrm{E}\left[\max I(X,F;Y)\right]=\mathrm{E}\left[ H(Y)\right]-\mathrm{E}\left[ H(Y|X,F)\right]=H(Y)-H(Y|X,F),
\label{equ_multicell_sumrate}
\end{align}
where $H()$ is the entropy function. The maximum operator in (\ref{equ_multicell_sumrate}) is removed because of the QAM symbol and the expectation operator is removed because the entropies of $Y$ and $N+I$ are already expectation values. Although closed-form integral of (\ref{equ_multicell_sumrate}) is not feasible because of the summation inside the logarithm function, upper bounds of the sum rate can be evaluated.

Since the entropy function $H()$ is concave, using Jensen's inequality, $H(Y|X,F)$ is lower bounded by

\begin{align}
H(Y|X,F)=H\left(\sum\limits_{k=1}^{Q}\omega_k\mathcal{CN}(z;0,\sigma_k^{2}) \right)&\geqslant \sum\limits_{k=1}^{Q}\omega_k H\left(\mathcal{CN}(z;0,\sigma_k^{2}) \right)\nonumber\\
&=\frac{1}{2}+\sum\limits_{k=1}^{Q}\omega_k\log_2 (2\pi e\sigma_k).
\label{equ_HYFX_LB}
\end{align}


On the contrary, $H(Y)$ is upper bounded by \cite{Huber08}
\begin{align}
H(Y)=H\left(\sum\limits_{k=1}^{Q}\omega^{\prime}_k\mathcal{CN}(y;0,\sigma_k^{\prime2}) \right) &\leqslant -\int \sum\limits_{k=1}^{Q}\omega^{\prime}_k\mathcal{CN}(y;0,\sigma_k^{\prime2})\log_2\left(\omega^{\prime}_k\mathcal{CN}(y;0,\sigma_k^{\prime2}) \right)dy\nonumber\\
&=\frac{1}{2}+\sum\limits_{k=1}^{Q}\omega^{\prime}_k\log_2 (2\pi e\sigma_k^{\prime}/\omega^{\prime}_k).
\label{equ_HY_UB}
\end{align}
To sum up, the multi-cell sum rate of OFDM-IM $r_3$ is upper bounded by
\begin{align}
 r_3\leqslant \sum\limits_{k^{\prime}=1}^{Q}\omega^{\prime}_{k^{\prime}}\log_2 (2\pi e\sigma_{k^{\prime}}^{\prime}/\omega^{\prime}_{k^{\prime}})-\sum\limits_{k=1}^{Q}\omega_k\log_2 (2\pi e\sigma_k).
\label{equ_r3_bound}
\end{align}
This upper bound provides insights in sum rate performance of multi-cell OFDM-IM.

\section{Results and Analysis} \label{sec_results_analysis}
Subcarrier index detection error probability and achievable rates of single cell OFDM-IM with different number of subcarriers are depicted in Fig. \ref{fig_Sidx_error_prob} and Fig. \ref{fig_Index_Modulation_sumrate}, respectively. In both figures, analytic results match simulated results well. It can be observed in Fig. \ref{fig_Index_Modulation_sumrate} that when the SNR is relatively low, the achievable rate contributed by subcarrier indexes is not significant. However, when SNR increases gradually, the gaps between achievable rates with different numbers of subcarriers first increase and then become stable.
\begin{figure} 
\centering\includegraphics[width=5in]{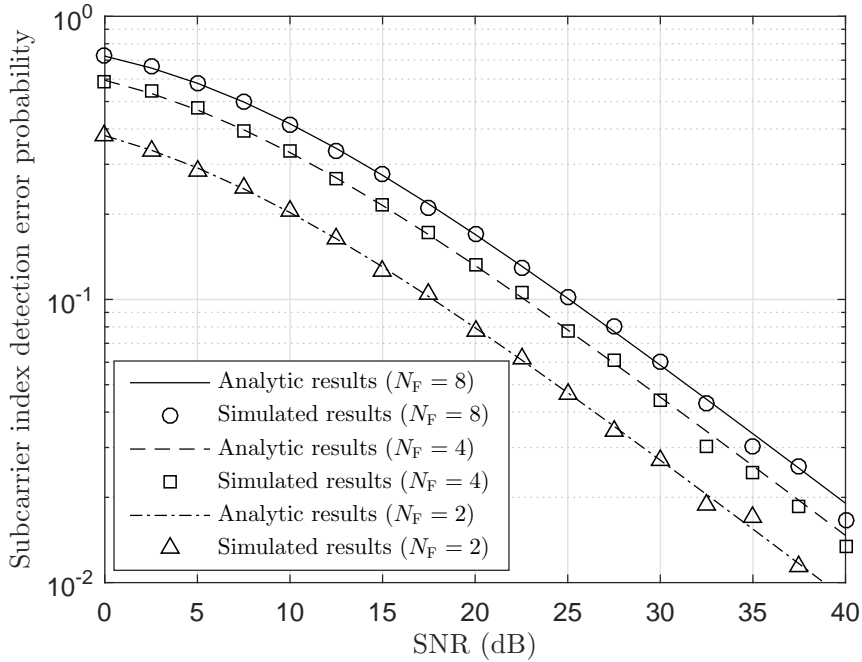} 
\centering\caption{Subcarrier index detection error probability of single cell OFDM-IM.}
\label{fig_Sidx_error_prob}
\end{figure}
\begin{figure} 
\centering\includegraphics[width=5in]{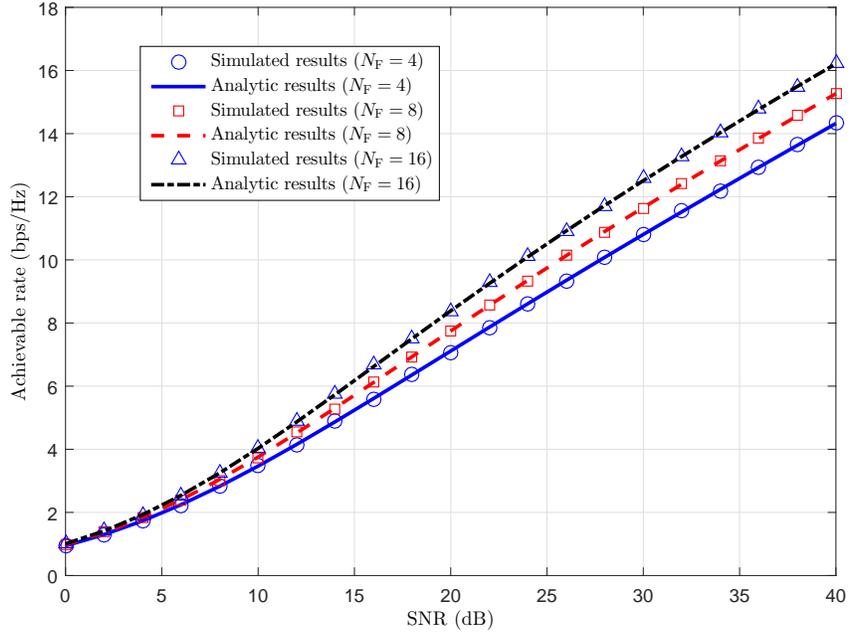} 
\centering\caption{Achievable rates of single cell OFDM-IM with different number of subcarriers.}
\label{fig_Index_Modulation_sumrate}
\end{figure}

Fig. \ref{fig_r2_Nfft} shows single cell OFDM-IM achievable rate contributed by subcarrier indexes with different numbers of subcarriers and different SNR values. From (\ref{equ_r2}), $r_2$ tends to zero when $N_\mathrm{F}$ tends to infinity. This suggests that the benefit of increasing $N_\mathrm{F}$ is diminishing and there is an optimal number of $N_\mathrm{F}$ such that $r_2$ reaches its peak. Closed-from expression of the optimal value of $N_\mathrm{F}$ is not available. However, this can be calculated by numerical results. It is shown in Fig. \ref{fig_r2_Nfft} that the optimal value of $N_\mathrm{F}$ increases with the SNR.
\begin{figure} 
\centering\includegraphics[width=5in]{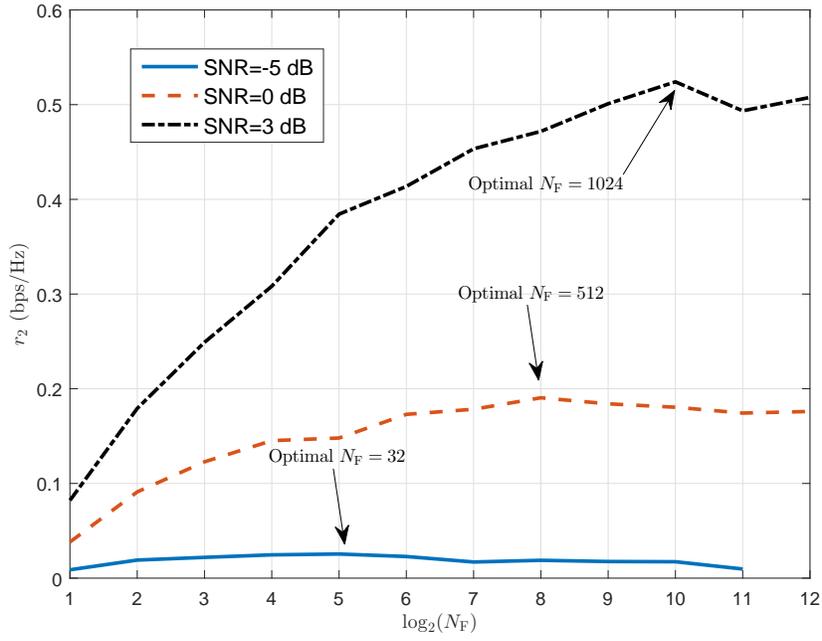} 
\centering\caption{The achievable rates contributed by different numbers of subcarriers and the optimal $N_\mathrm{F}$.}
\label{fig_r2_Nfft}
\end{figure}

In multi-cell simulations, thermal noise with power density  $-173$ dBm/Hz \cite{36814} is assumed in the multi-cell network. Also, the bandwidth of each subcarrier is 15 kHz \cite{36814}. In this case, the noise power per subcarrier can be calculated by $\sigma_\mathrm{N}^2=7.5\times 10^{-11}$ W. 

Fig. \ref{fig_IM_CDF} illustrates the CDFs of SINR in terms of different values of $N_\mathrm{F}$ in the multi-cell OFDM-IM scenario. The SINR is larger when $N_\mathrm{F}$ is larger, because the target UE has a smaller probability of being interfered, where ICI is more spreaded in the frequency domain. In addition, the simulated results reasonably well align with analytic results.

\begin{figure} 
\centering\includegraphics[width=5in]{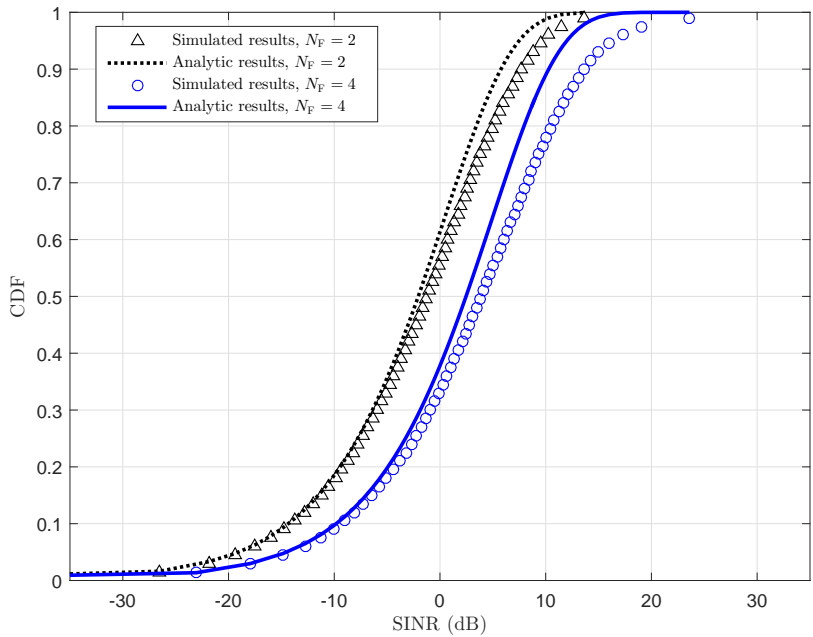} 
\centering\caption{Comparison of CDFs of SINR with respect to different numbers of subcarriers of multi-cell OFDM-IM ($\lambda=10^{-4}$, $\alpha=3$, $P_T=40$W, $d=50$m).}
\label{fig_IM_CDF}
\end{figure}

The PDF of the real part of noise plus ICI of multi-cell OFDM-IM with 4QAM inputs is illustrated in Fig. \ref{fig_FQAM_Interf_PDF}. The total number of BSs is $19$ and the inter site distance (ISD) is 100m. It be can observed that the MoG distribution derived in this paper fits simulation excellently.
The generalized Gaussian model proposed in \cite{Hong14}\cite{Seol09} and the Gaussian model are also shown in Fig. \ref{fig_FQAM_Interf_PDF}. However, these two models fail to match the realistic noise plus ICI well.
 The generalized Gaussian model is able to capture the peak while it does not model the spread well. On the other hand, the Gaussian model has better alignment with the PDF spread than the generalized Gaussian model, but it is not accurate to model the peak.

\begin{figure} 
\centering\includegraphics[width=5in]{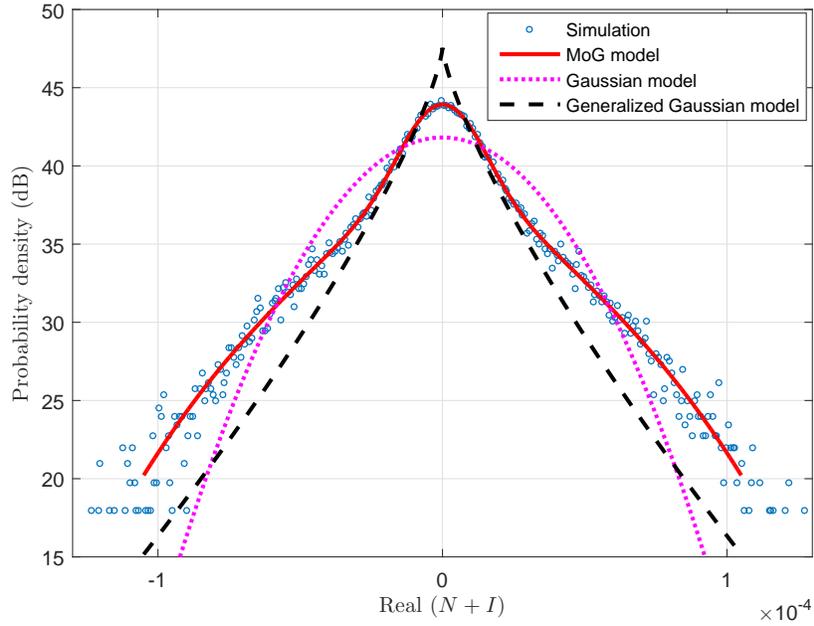} 
\centering\caption{PDF of the real part of noise plus ICI of multi-cell OFDM-IM with 4QAM inputs ($N_\mathrm{B}=19$, $\alpha=3$, $P_T=40$W, $d=50$m, $Q'=4$).}
\label{fig_FQAM_Interf_PDF}
\end{figure}

The upper bound of sum rates of multi-cell OFDM-IM with 4QAM is depicted in Fig. \ref{fig_Index_Modulation_multicell_sumrate}. Single cell achievable rates with Gaussian input and 4QAM are included as reference. It can be observed that the sum rate of multi-cell OFDM-IM is at least approximately 20\% worse than single cell results because of ICI. When the SNR is smaller, although the multi-cell result outperforms the other two single cell results, this is caused by the upper limit. 
\begin{figure} 
\centering\includegraphics[width=5in]{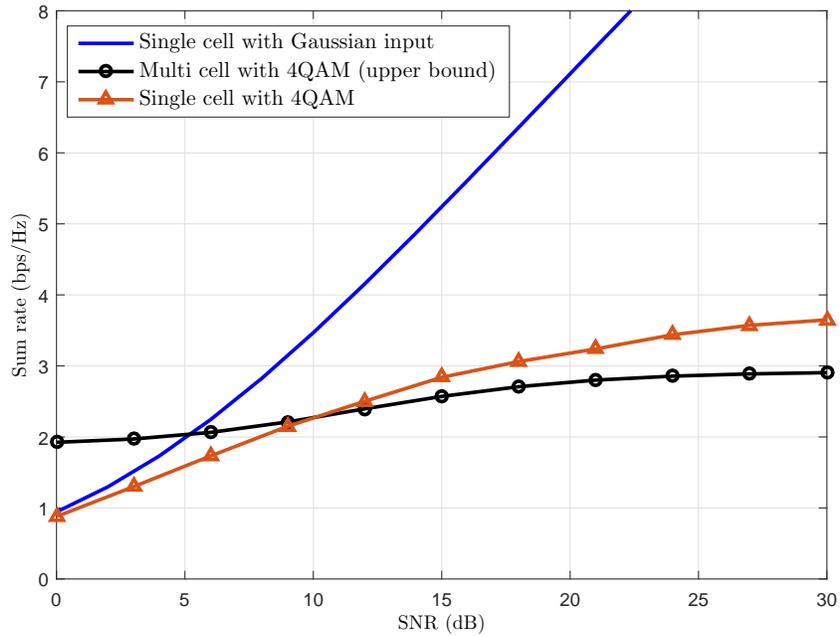} 
\centering\caption{Upper bound of sum rates of multi-cell OFDM-IM ($N_\mathrm{B}=19$, $N_\mathrm{F}=4$).}
\label{fig_Index_Modulation_multicell_sumrate}
\end{figure}

\section{Conclusions} \label{sec_conclusion_section}
In this paper, single cell achievable rate and multi-cell statistical properties of SINR, ICI, and sum rate for OFDM-IM have been studied. It has been shown that the increase of the number of subcarriers does not improve the single cell achievable rate in the low SNR regime. The main benefit of more subcarriers appears in multi-cell scenarios, where less BSs will be interfering the target UE. In addition, the PDF of noise plus ICI has been derived for OFDM-IM with QAM input, showing that noise plus ICI follows a MoG distribution. The parameters of the MoG distribution have been estimated by a simplified EM algorithm in this paper. Later, upper bound of sum rate of multi-cell OFDM-IM has been derived, which can be used as performance guideline for OFDM-IM networks. For future work, it will be practical to use the PDF of noise plus ICI to develop multi-cell OFDM-IM signal detection algorithms. Also, the analysis method in this paper can be extended to $N_{\mathrm{F}}$-ary asymmetric channel to investigate achievable rates of generalized OFDM-IMs.

\section{Acknowledgement}
This work has been performed in the framework of the
Horizon 2020 project FANTASTIC-5G (ICT-671660) receiving
funds from the European Union. The authors would like
to acknowledge the contributions of their colleagues in the
project, although the views expressed in this contribution are
those of the authors and do not necessarily represent the
project.

%

\end{document}